\RequirePackage{fix-cm}
\documentclass[smallextended,envcountsect,envcountsame,natbib]{svjour3}
\smartqed

\usepackage[T1]{fontenc}
\usepackage{amsmath,amssymb}
\usepackage{graphicx}
\usepackage{booktabs}
\usepackage{array}
\usepackage{enumitem}
\usepackage{tikz}
\usepackage{xcolor}
\usepackage{xspace}
\usepackage[hidelinks]{hyperref}
\newcolumntype{L}[1]{>{\raggedright\arraybackslash}p{#1}}

\newcommand{\MCMKP}{\textnormal{MinMax-Pack}\xspace}
\newcommand{\MPMKC}{\textnormal{MaxMin-Cov}\xspace}
\newcommand{\LBP}{\textnormal{\textsc{LBP}}\xspace}
\newcommand{\GBP}{\textnormal{\textsc{GBP}}\xspace}
\newcommand{\KP}{\textnormal{\textsc{KP}}\xspace}
\newcommand{\SubsetSum}{\textnormal{\textsc{Subset Sum}}\xspace}
\newcommand{\OPT}{\operatorname{OPT}}
\newcommand{\optP}{\OPT_{\mathrm{P}}}
\newcommand{\optC}{\OPT_{\mathrm{C}}}
\newcommand{\NP}{\mathsf{NP}}
\newcommand{\PP}{\mathsf{P}}
\newcommand{\APX}{\mathsf{APX}}
\newcommand{\XP}{\mathsf{XP}}
\newcommand{\FPT}{\mathsf{FPT}}
\newcommand{\Wone}{\mathsf{W[1]}}
\newcommand{\vseven}[1]{\textcolor{black}{#1}}

\journalname{Journal of Combinatorial Optimization}

\begin{document}

\title{Inclusion-wise extremal knapsack problems are (mostly) hard to approximate}
\titlerunning{Inclusion-wise extremal knapsack problems}

\author{Fabio Ciccarelli \and Fabio Furini \and Ulrich Pferschy \and Elena Rener}
\authorrunning{F. Ciccarelli, F. Furini, U. Pferschy, E. Rener}

% JOCO asks for department, institution, city and country of every author,
% and for a corresponding author (to be chosen).  Still missing: the
% corresponding author.
\institute{F. Ciccarelli \at
             Department of Computer, Control and Management Engineering
             ``Antonio Ruberti'' (DIAG),
             Sapienza University of Rome, Rome, Italy\\
             \email{f.ciccarelli@uniroma1.it} \\
             ORCID: 0009-0005-5795-5574
           \and
           F. Furini \at
             Department of Computer, Control and Management Engineering
             ``Antonio Ruberti'' (DIAG),
             Sapienza University of Rome, Rome, Italy\\
             \email{fabio.furini@uniroma1.it}\\
             ORCID: 0000-0002-1839-5827
           \and
           U. Pferschy \at
             Department of Operations and Information Systems,
             University of Graz, Graz, Austria\\
             \email{ulrich.pferschy@uni-graz.at}\\
             \vseven{ORCID: 0000-0001-8881-1497}
           \and
           E. Rener \at
             Independent researcher, Bern, Switzerland\\
             \email{elenarener@gmail.com}\\
             ORCID: 0000-0002-2371-8360
}

\date{Received: date / Accepted: date}

\maketitle

\begin{abstract}
\vseven{We study two inclusion-wise extremal variants of the 0--1 Knapsack
Problem (\KP).}
Given $n$ items with integer weights and nonnegative costs or profits, the
Minimum-Cost Maximal Knapsack Packing Problem seeks a minimum-cost packing
for capacity $c$ to which no item can be added, while the Maximum-Profit Minimal Knapsack
Cover Problem seeks a maximum-profit cover for threshold $d$ from
which no item can be removed.  Both problems are weakly $\NP$-hard
and pseudo-polynomially solvable, as is the classical \KP, but their
approximability is fundamentally different: \vseven{unless $\PP=\NP$,
with positive values neither problem admits a polynomial-time
$2^{p(n)}$-approximation for any polynomial $p$; allowing zero values
rules out every multiplicative factor.} \vseven{Even} with unit costs for packing and
unit profits for covering, neither
admits an $n^{1-\varepsilon}$-approximation for any fixed $\varepsilon>0$,
whereas a trivial $n$-approximation exists.
\vseven{We identify restrictions that restore approximability.}
If every item weighs at least a fixed fraction of the capacity,
respectively of the threshold, both problems can be solved exactly in
polynomial time.  For the cover problem, if each weight $w_i$ satisfies
$w_i-1\ge\delta(d-1)$, we give an approximation scheme for arbitrary
profits whose running time is polynomial in $n$, $1/\varepsilon$, and
$1/\delta$.  If values equal weights, a fully polynomial-time approximation
scheme (FPTAS) was already known for the packing problem; we give an FPTAS
for the cover problem, improving on the previously known factor $2$.
\vseven{If value-to-weight ratios differ by at most a factor $\kappa$, we
obtain a $(1+\varepsilon)\kappa$-approximation for both problems.}
\end{abstract}

% ===========================================================================
\section{Introduction}
\label{sec:intro}
% ===========================================================================

The 0--1 Knapsack Problem (\KP) is the standard example of a weakly $\NP$-hard
problem that is very well behaved.  With $n$ items and capacity $c$ it is
solved exactly in $O(nc)$ time by dynamic programming \citep{KPP2004}, and
it admits a fully polynomial-time approximation scheme (FPTAS), that is, for
every $\varepsilon>0$ an algorithm that returns a solution within a factor
$1+\varepsilon$ of the optimum in time polynomial in the input size and $1/\varepsilon$
\citep{IbarraKim1975,Lawler1979,KellererPferschy2004}.  It is obtained by scaling
the profits and running a profit-indexed dynamic program.  For \KP,
pseudo-polynomial exact solvability and full approximability go together,
and the second follows from the first, through the profit-indexed
version of the dynamic program.  The central question of this paper
is whether this implication survives when feasible solutions are required
to be inclusion-wise extremal.

This paper studies two variants of \KP in which feasible solutions are
required to be \emph{inclusion-wise extremal}.  We are given a set $I$ of
$n$ items; item $i\in I$ has a positive integer weight $w_i$ and a
nonnegative value $v_i$, and we write $w(S)=\sum_{i\in S}w_i$ for
$S\subseteq I$; in particular, $w(I)$ is the total weight of all items.
In the Minimum-Cost Maximal Knapsack Packing Problem
(\MCMKP) the values are costs.  One must select a set of items that fits
into a knapsack of positive integer capacity $c$ and to which no
further item can be added without exceeding the capacity, and among these
\emph{maximal} packings one of minimum cost is sought.  In the
Maximum-Profit Minimal Knapsack Cover Problem (\MPMKC) the values are
profits.  One must select a set of items whose
total weight reaches a positive integer threshold $d$ and from which no item can be removed without dropping
below it, and among these \emph{minimal} covers one of maximum profit is
sought.  Throughout the paper, the weights $w_i$, the capacity $c$ and the
threshold $d$ are positive integers.  The two problems are complementary
\citep{FLS2017}: the complement
of a maximal packing is a minimal cover for the threshold $d=w(I)-c$, and vice
versa.  Both problems are weakly
$\NP$-hard and are solved by the same kind of dynamic program as \KP: \MCMKP in
$O(nc)$ time, and \MPMKC in $O(nd)$ time directly or in $O(n(w(I)-d))$ time
by complementation.

Both problems come from scheduling.  Let every item be a job with processing
time $w_i$, released at time $0$ and due at time $c$, and consider a single
worker who is not allowed to stay idle while some job can still be started
and completed.  The sets of jobs the worker may execute are exactly the maximal packings.
Minimising the total processing time is the Lazy
Bureaucrat Problem (\LBP) with common arrivals and deadlines \citep{Arkin2003,GMP2013}.
Minimising the number of executed jobs is the
unit-cost case.  A cost $v_i$ unrelated to $w_i$ models a worker who dislikes
some jobs more than others, and minimising the total disutility of the
executed jobs is exactly \MCMKP.
The cover problem is the maximisation
counterpart: when $v_i=w_i$ it is the Greedy Boss Problem (\GBP) of \citet{GMP2013}.

The answer to the approximability of the studied knapsack problems is surprisingly negative: inclusion-wise extremality preserves
pseudo-polynomial solvability, but can almost completely destroy
approximability.  In the
proportional case $v_i=w_i$ approximability survives: the Lazy Bureaucrat
Problem with common arrivals and deadlines admits an FPTAS \citep{GMP2013}.
In general, however, although \MCMKP and \MPMKC retain the weak
$\NP$-hardness and the pseudo-polynomial exact algorithms of \KP, the
approximation picture is radically different when the objective is allowed
to vary independently of the weights.  \vseven{With positive values,
neither problem admits an approximation within $2^{p(n)}$ for any
polynomial $p$; allowing zero values rules out every multiplicative factor.}
\vseven{With} unit costs for \MCMKP or unit profits for \MPMKC,
even a factor $n^{1-\varepsilon}$ is impossible, whereas every feasible
solution is within a factor $n$ of the optimum.
Table~\ref{tab:comparison} summarises this separation.

\begin{table}[ht]
\caption{Classical knapsack versus general \MCMKP and \MPMKC, assuming
$\PP\neq\NP$.}
\label{tab:comparison}
\centering\small
\begin{tabular}{@{}L{0.26\linewidth}L{0.06\linewidth}L{0.28\linewidth}L{0.28\linewidth}@{}}
\toprule
 & \KP & \MCMKP & \MPMKC \\
\midrule
Weakly $\NP$-hard & yes & yes (Thm.~\ref{thm:fls}) & yes (Thm.~\ref{thm:fls}) \\
Weight-indexed pseudo-polynomial DP & yes & yes (Thm.~\ref{thm:fls}) & yes (Thm.~\ref{thm:fls}) \\
Value-indexed pseudo-polynomial DP & yes & no (Cor.~\ref{cor:nofptas}) & no (Cor.~\ref{cor:nofptas}) \\
FPTAS & yes & no, even with unit costs (Cor.~\ref{cor:nofptas}) & no, even with unit profits (Cor.~\ref{cor:nofptas}) \\
Constant-factor approximation & yes & no, even with unit costs (Thm.~\ref{thm:unitgap}) & no, even with unit profits (Thm.~\ref{thm:unitgap}) \\
$n^{1-\varepsilon}$-approximation & yes & no, even with unit costs (Thm.~\ref{thm:unitgap}) & no, even with unit profits (Thm.~\ref{thm:unitgap}) \\
\vseven{Factor $2^{p(n)}$ with positive values} & yes & \vseven{no (Thm.~\ref{thm:noconst})} & \vseven{no (Thm.~\ref{thm:noconst})} \\
\bottomrule
\end{tabular}
\end{table}

Our results also identify the two structural features behind this loss of
approximability: items of vanishing relative weight and objectives that are
poorly correlated with the weights.  This leads naturally to the second
question of the paper: what happens when either of these two features is
ruled out?  We show that ruling out either feature allows a substantially better algorithmic behaviour.
On the weight side,
excluding arbitrarily small items yields exact polynomial-time algorithms.
On the objective side, proportional objectives admit approximation schemes,
in particular an FPTAS for the Greedy Boss Problem, for which only a
factor $2$ was known.
Thus, at the proportional endpoint both problems
admit FPTASs, and bounding the value-to-weight spread still yields constant-factor  approximations.
Table~\ref{tab:landscape}
at the end of the paper summarises the picture.

% ===========================================================================
\section{Our contributions}
\label{sec:contrib}
% ===========================================================================

The picture of approximation that we will develop has two sides.
We first establish how severe the
loss of approximability can be, and then identify structural restrictions
under which efficient approximation---or even exact polynomial-time
solvability---is recovered.  The main hardness results hold for both problems of the pair, while the
positive results also reveal some asymmetries between the packing and cover sides.  The packing and cover problems are complementary
at the level of feasible solutions, but not at the level of approximation
ratios, since complementation is affine on the objective
(Section~\ref{sec:complement}); hardness gaps, however, do transfer.

Our first two results concern inapproximability and show how poorly the
problems can be approximated in general.
\begin{enumerate}[label=(\roman*),leftmargin=2.5em,itemsep=2pt]
  \item \vseven{\emph{Inapproximability with positive values}}
        (Theorem~\ref{thm:noconst}).  \vseven{Unless $\PP=\NP$, even with
        positive values neither \MCMKP nor \MPMKC admits a polynomial-time
        $2^{p(n)}$-approximation for any polynomial $p$.  If zero values
        are allowed, no multiplicative approximation factor is possible.}
        This closes the case
        left open between \citet{Arkin2003}, \citet{Esfahbod2003} and
        \citet{GaiZhang2008}.
  \item \emph{No $n^{1-\varepsilon}$-approximation, even with unit
        costs or profits}
        (Theorem~\ref{thm:unitgap}, Proposition~\ref{prop:trivial}).  For
        every fixed $0<\varepsilon<1$, unit-cost \MCMKP and unit-profit
        \MPMKC admit no polynomial-time $n^{1-\varepsilon}$-approximation
        unless $\PP=\NP$, while any maximal packing, or minimal cover, is an
        $n$-approximation, so the exponent one is tight.  The construction
        is the one used by \citet{GaiZhang2008} (Theorem~7 there), with the number of
        tiny jobs raised from $\tilde n+1$ to $\tilde n^k$, for a constant
        $k>1/\varepsilon$; our contribution is the observation that the
        resulting gap rules out every ratio $n^{1-\varepsilon}$ with
        $\varepsilon>0$ fixed, on both sides of the pair.
\end{enumerate}
Both reductions rest on \emph{exact filling} (Section~\ref{sec:struct}):
a unit-weight item outside a maximal packing, or inside a minimal cover,
forces the load to equal the capacity, or the threshold, exactly.

The next three results show how approximability is recovered by
restricting the input data, first the weights and then the values.
\begin{enumerate}[label=(\roman*),leftmargin=2.5em,itemsep=2pt,resume]
  \item \emph{No small items} (Theorem~\ref{thm:nonsmall},
        Proposition~\ref{prop:fptas-delta}).  If every item weighs at least
        a fixed fraction $\delta$ of the capacity, or of the threshold, both
        problems are solved exactly in $n^{O(1/\delta)}$ time.  On the cover
        side the dependence on $\delta$ can be made polynomial: under the
        slightly stronger condition $w_i-1\ge\delta(d-1)$, \MPMKC admits,
        for arbitrary profits, an approximation scheme whose running time
        is polynomial in $n$, $1/\varepsilon$ and $1/\delta$.
  \item \emph{Approximation schemes for the Greedy Boss Problem}
        (Corollary~\ref{cor:ptas}, Theorem~\ref{thm:gbpfptas}).  At the
        proportional endpoint $v_i=w_i$ the packing problem was known to admit an FPTAS
        \citep{GMP2013}, while for the cover problem only the tight factor
        $2$ of Max Greedy was known.
        An algorithm that guesses the $k$
        heaviest items of an optimal minimal cover and completes the solution greedily is a $\frac{k+2}{k+1}$-approximation in $O(n^{k+1})$ time; $k=1$
        already improves the factor $2$ to $3/2$.  The technique is the
        cover-side counterpart of the PTAS of \citet{GMP2014} for the lazy
        matroid problem.  Most importantly, we exploit the two-sided residual
        structure, together with an FPTAS for \SubsetSum, to obtain an FPTAS
        for \GBP.  This
        settles the question of \citet{GMP2013} and shows that at the
        proportional endpoint both problems admit FPTASs.
  \item \emph{Bounded spread} (Theorem~\ref{thm:spread}).  Extending the proportional case, if the ratio
        $\kappa$ between the largest and the smallest value-to-weight ratio
        is bounded, both problems admit a
        $(1+\varepsilon)\kappa$-approximation for every
        $\varepsilon>0$, in time polynomial in the input size and $1/\varepsilon$.
\end{enumerate}

\paragraph{Organisation of the paper.}
Section~\ref{sec:lit} places our results in the literature.
Section~\ref{sec:def} develops the structural properties used throughout.
Section~\ref{sec:hard} proves the inapproximability results and explains
why the classical profit-scaling route to an FPTAS fails.
Sections~\ref{sec:pos} and~\ref{sec:values} show how approximability is
recovered by restricting the weights and the values, respectively; the
latter contains the approximation schemes for the Greedy Boss Problem.  Section~\ref{sec:concl} concludes with the resulting approximation
landscape and open questions.

% ===========================================================================
\section{Literature review}
\label{sec:lit}
% ===========================================================================

Maximal packings and minimal covers have been studied in three communities
that rarely cite each other.  In scheduling, they appear as the feasible sets
of the Lazy Bureaucrat Problem, where a worker who wants to do as little work as
possible is not allowed to stay idle.  In the knapsack literature, they appear
as knapsack problems with an additional maximality or filling constraint,
studied mainly from the point of view of exact algorithms.  In bin packing,
the same ``reverse'' viewpoint has produced several lazy variants with many
bins.  We review the three lines in turn, then recall what is known on the
approximation of the classical knapsack problem, which is the benchmark
against which our results should be read.  Section~\ref{sec:known}
summarises the state of the art for the two problems studied here and
identifies the questions left open.

\subsection{The Lazy Bureaucrat problem}

\citet{Arkin2003} introduced the Lazy Bureaucrat Problem \LBP: a single worker
must process jobs with release times, processing times and deadlines, may
not stay idle while an executable job is available, and wants to do as
little work as possible.  They consider three objectives, total time spent,
weighted number of executed jobs and makespan, and prove that the
nonpreemptive problem is strongly $\NP$-hard and not approximable within any
fixed factor.  With a common release time the problem becomes weakly $\NP$-hard,
and they claim a pseudo-polynomial dynamic program (Theorem~9 there)
without giving the algorithm.  Their inapproximability
construction (Theorem~1 there) uses a common release time but a job with a
\emph{different deadline}.  \citet{HepnerStein2002} study the preemptive
version.

\citet{Esfahbod2003} consider a common deadline and arbitrary release times.
They prove weak $\NP$-hardness for all objectives, give a pseudo-polynomial
dynamic program, show that Shortest Job First is a tight $2$-approximation
for the makespan, and prove that minimising the number of executed jobs, and
hence a weighted sum of executed jobs, is not approximable within any fixed
factor (Theorem~2 and Corollary~1 there).  Their construction uses extra
jobs with a \emph{different release time}.

\citet{GaiZhang2008}, extending \citet{GaiZhang2006}, stay with a common
deadline.  They show that Shortest Job First has ratio $2$ also for the
total time spent, and give polynomial-time approximation schemes for the
makespan and for the time spent; the scheme for the makespan guesses the $k$
largest jobs of an optimal schedule and completes greedily, with ratio
$1+1/k$.  Their Theorem~7 proves that the problem remains weakly $\NP$-hard
under all four objectives when the release time \emph{and} the deadline are
common, which is exactly the setting of the present paper.  For the
number-of-jobs objective the reduction adds, to an instance of \SubsetSum
(decide whether some of $\tilde n$ given positive integers sum exactly to a
given target), $\tilde n+1$ tiny jobs of length $1/(\tilde n+1)$; the
answer is yes if and only if at most $\tilde n$ jobs are executed.  No inapproximability bound is stated for this setting.

\citet{GMP2013} study common arrivals and deadlines with the total time
spent as objective, that is, \MCMKP with $v_i=w_i$.  They give a linear-time
$4/3$-approximation and an FPTAS, and study strategyproof mechanisms w.r.t.\ job duration.
They also introduce the maximisation counterpart, the Greedy Boss Problem \GBP, and prove that it is $\NP$-complete.
It is also shown that the longest-job-first greedy algorithm Max Greedy is a $1/2$-approximation, more precisely a $t/(t+1)$-approximation when it returns $t$ jobs.  In their concluding section they ask for ``a
complete analysis of the Greedy Boss Problem \ldots in terms of
approximability''.  The same authors generalise the Lazy Bureaucrat Problem
to matroids in \citet{GMP2014} and give a PTAS for it, still with the
objective equal to the weight; the scheme guesses the $k$ heaviest elements
of an optimal solution and has ratio $\frac{k+1}{k}$.  \citet{GaiZhang2017} study an online
version with a machine deadline.

\subsection{Knapsack models with maximality and filling constraints}

\citet{FLS2015} study the packing problem with common arrivals and deadlines and
arbitrary costs, still under the name Lazy Bureaucrat Problem: they give the
first pseudo-polynomial dynamic program for it, based on the decomposition by
the smallest excluded item, together with integer and constraint programming
formulations.
\citet{FLS2017} introduce the names \textsc{MCMKP} and \textsc{MPMKC} for
the two problems with costs, or profits, independent of the weights, prove
the complementarity between them, and give the $O(nc)$ dynamic program
together with a computational study.  They deliberately refrain from the
names Lazy Bureaucrat and Greedy Boss when the objective is not
proportional to the weights, and we follow them: throughout, \LBP and \GBP
denote the proportional restrictions, and \MCMKP and \MPMKC the general
problems.  Recently \citet{MPS2026} met the same minimum-filling subproblem
inside that dynamic program in a fairness-constrained knapsack problem.

\citet{Xu2013}, extending \citet{XuLai2011}, studies the knapsack problem with a minimum filling
constraint: for a given $\beta\in[0,1]$, the total weight of the selected
items must be at least $\beta c$.  For $\beta=1$ deciding feasibility is
already \SubsetSum, hence $\NP$-complete; for $\beta<1$ an approximation scheme is given that returns a feasible solution and runs in time polynomial in $n$, $1/\varepsilon$ and $1/(1-\beta)$.
Maximality forces a minimum load, so every maximal packing
is feasible for a minimum-filling instance, but the converse fails; the exact
correspondence, which we use in Section~\ref{sec:fpt}, is with the residual
problem obtained after fixing the smallest selected item of a minimal
cover.  The survey of single knapsack problems by \citet{Cacchiani2022}
mentions \MCMKP only through \citet{FLS2017}, and lists the
minimum-filling knapsack among the less studied variants.

\subsection{Related packing and covering problems}

Other ``lazy'' packing problems have been studied with several bins and a cardinality objective.
In lazy bin covering every bin must be a minimal
cover of the capacity, at most one bin may be partially filled, and the number of bins is minimised; it admits an asymptotic PTAS
\citep{LinYangXu2010} and is $\APX$-hard \citep{GaiZhang2009}.  Neither the
model nor the results are comparable with ours.  A single bin of the Open-End Bin Packing Problem
\citep{LeungDrorYoung2001,YangLeung2003,Balogh2021} has a structure
reminiscent of a minimal cover, since one designated last item may cross the
threshold, but there the objective is the number of bins and only that item
must be removable.  \citet{GMT2018} add digraph constraints to subset sum; two of
their four problems impose maximality and generalise the \LBP, again with
the objective equal to the weight.  They prove strong $\NP$-hardness on
directed acyclic graphs and give a PTAS on them.

% The graph problems with an inclusion-wise extremal objective (max min
% vertex cover, upper domination, max min feedback vertex set, max minimal
% set cover) share only the shape of the objective; they are strongly
% NP-hard and their techniques do not transfer.  The team decided on
% 9 September 2026 not to cite them.  If a referee asks, one sentence with
% \citet{Dublois2021,DubloisLampisPaschos2022,Damasceno2026} is enough.

\subsection{Approximation of the knapsack problem}

For \KP, \citet{Sahni1975} gave a PTAS and \citet{IbarraKim1975} the first FPTAS.
The running time has since been improved many times
\citep{Lawler1979,KellererPferschy2004,ChenLianMaoZhang2024,Mao2024}.
The minimisation knapsack problem, minimise the cost subject to a lower bound on the load,
also admits an FPTAS.
Approximation guarantees do not transfer between it
and \KP \citep{GuntzerJungnickel2000}, but a constant-factor bound from
their greedy algorithms followed by cost scaling and a cost-indexed dynamic
program works, because the load constraint is one-sided.  All
these schemes rest on the same dominance property: among partial solutions
of equal profit the lightest one is best.  This contrast is central to our
paper: maximality preserves the weight-indexed pseudo-polynomial dynamic
program, but destroys the profit-side dominance that underlies classical
scaling arguments (cf.~Section~\ref{sec:hard}).

\subsection{Summary of known results}
\label{sec:known}

Table~\ref{tab:known} places the above results from a scheduling context along the two restrictions that define our problems, common release and common deadline, and along the
objective.  The table makes the missing case explicit:
the two known inapproximability results rely either on different release
times or on different deadlines, whereas with common release, common
deadline and independent costs, weak $\NP$-hardness and pseudo-polynomial
solvability were known, but no approximation guarantee or inapproximability
bound was available.
Approximability in this case, and the approximability of the Greedy Boss
Problem beyond the factor $2$, are the gaps this paper fills.

\begin{table}[ht]
\caption{Closest known results and the approximation gap addressed in this paper.}
\label{tab:known}
\centering\small\setlength{\tabcolsep}{4pt}
\begin{tabular}{@{}L{0.19\linewidth}L{0.11\linewidth}L{0.11\linewidth}L{0.19\linewidth}L{0.27\linewidth}@{}}
\toprule
Reference & Release & Deadline & Objective & Result \\
\midrule
\citet{Arkin2003} & common & different & time spent, weighted number, makespan & no fixed factor; weakly $\NP$-hard; pseudo-polynomial DP claimed \\
\citet{Esfahbod2003} & different & common & number of jobs, weighted sum & no fixed factor \\
\citet{GaiZhang2008} & common & common & all four & weakly $\NP$-hard only \\
\citet{GMP2013} & common & common & time spent ($v_i=w_i$) & $4/3$ in linear time; FPTAS \\
\citet{GMP2013} & --- & --- & \GBP ($v_i=w_i$, cover) & $\NP$-complete; factor $2$ \\
\citet{FLS2017} & common & common & independent costs & $O(nc)$ DP; complementarity \\
\midrule
this paper, Thm.~\ref{thm:noconst} & common & common & independent costs & \vseven{no factor $2^{p(n)}$ with positive costs; no factor if zero is allowed} \\
this paper, Thm.~\ref{thm:unitgap}, Prop.~\ref{prop:trivial} & common & common & unit costs / profits & no $n^{1-\varepsilon}$; factor $n$ by any feasible solution \\
this paper, Thm.~\ref{thm:gbpfptas} & --- & --- & \GBP & FPTAS \\
\bottomrule
\end{tabular}
\end{table}

% ===========================================================================
\section{Definitions and preliminaries}
\label{sec:def}
% ===========================================================================

The results that follow rely on three elementary structural facts:
maximality and minimality are determined by a single extremal item, a
unit-weight item may force an exact filling, and packing and cover solutions
are complementary.  We establish these facts before turning to
approximation, and close the section with results from the literature that we use as black boxes.

\subsection{The two problems under consideration}

Let $I=\{1,\ldots,n\}$ be a set of items.  Item $i$ has a positive integer
weight $w_i$ and a nonnegative value $v_i$, which is a cost on the packing
side and a profit on the cover side.  For $P\subseteq I$ write
$w(P)=\sum_{i\in P}w_i$ and $v(P)=\sum_{i\in P}v_i$.  Uppercase letters are
sets and lowercase letters are scalars; $P$ always denotes a packing and $C$
a cover, while $c$ is the capacity and $d$ the threshold.  The words
\emph{maximal} and \emph{minimal} refer to set inclusion.  Capacities $c$
and thresholds $d$ are positive integers.  A set $P\subseteq I$ is a
\emph{packing} for $c$ if $w(P)\le c$, and a set $C\subseteq I$ is a
\emph{cover} for $d$ if $w(C)\ge d$.

\begin{definition}[Maximal packing, \MCMKP]
\label{def:mcmkp}
Let $c$ be a positive integer.  A packing $P$ for $c$ is \emph{maximal} if $w(P)+w_j>c$ for all
$j\in I\setminus P$.  The
problem \MCMKP asks for
$\optP=\min\{v(P): P \text{ is a maximal packing for } c\}$.
\end{definition}

\begin{definition}[Minimal cover, \MPMKC]
\label{def:mpmkc}
Let $d$ be a positive integer.  A cover $C$ for $d$ is \emph{minimal} if
$w(C\setminus\{j\})<d$ for all $j\in C$.  The problem
\MPMKC asks for
$\optC=\max\{v(C): C \text{ is a minimal cover for } d\}$.
\end{definition}

\paragraph{Nontrivial instances.}
On the packing side we assume $\min_{i\in I}w_i\le c<w(I)$, so that neither
$\emptyset$ nor $I$ is a maximal packing; on the cover side we assume
$0<d\le w(I)$.  Other instances are recognised and solved in linear time.

\subsection{Special cases and approximation terminology}

The \emph{spread} $\kappa$ of an instance with positive values is
\begin{equation}
\label{eq:kappa}
 \kappa=\frac{\max_{i\in I}v_i/w_i}{\min_{i\in I}v_i/w_i}\ \ge 1,
\end{equation}
equivalently, $\kappa=\lambda_+/\lambda_-$ where
$\lambda_+=\max_{i\in I}v_i/w_i$ and
$\lambda_-=\min_{i\in I}v_i/w_i$, so that
$\lambda_-w_i\le v_i\le\lambda_+w_i$ for all $i\in I$.  The proportional case is $\kappa=1$; scaling
shows that we may then assume $v_i=w_i$.  Table~\ref{tab:restrictions}
lists the restrictions studied.

\begin{table}[ht]
\caption{Restrictions of the two problems.}
\label{tab:restrictions}
\centering\small\setlength{\tabcolsep}{4pt}
\begin{tabular}{@{}lll@{}}
\toprule
Restriction & Packing side (\MCMKP) & Cover side (\MPMKC) \\
\midrule
Unit costs/profits & $v_i=1$ for all $i\in I$ & $v_i=1$ for all $i\in I$ \\
Proportional & $v_i=w_i$ for all $i\in I$: the \LBP & $v_i=w_i$ for all $i\in I$: the \GBP \\
Bounded spread & $\kappa$ fixed, see \eqref{eq:kappa} & $\kappa$ fixed \\
Non-small items & $w_i\ge\delta c$ for all $i\in I$, $\delta$ fixed &
                  $w_i\ge\delta d$ for all $i\in I$, $\delta$ fixed \\
\bottomrule
\end{tabular}
\end{table}

For $\alpha\ge1$, an $\alpha$-approximation returns a feasible solution of
value at most $\alpha\cdot\OPT$ for a minimisation problem and at least
$\OPT/\alpha$ for a maximisation problem; a maximisation guarantee usually
written ``$1/2$-approximation'' is a factor-$2$ approximation in this
convention.  A PTAS is a family of polynomial-time
$(1+\varepsilon)$-approximations, one for every fixed $\varepsilon>0$; an
FPTAS is a PTAS polynomial in the input size and $1/\varepsilon$ jointly.  In hardness
statements of the form ``no $n^{1-\varepsilon}$-approximation'',
$\varepsilon$ is a fixed constant in $(0,1)$.  Inside reductions, $\tilde n$
denotes the number of integers of the \SubsetSum instance and $n$ the number of items of
the constructed instance.

\subsection{Structural properties}
\label{sec:struct}

Checking maximality or minimality does not require looking at every item:
one item suffices in each case.  This simple observation drives both the
exact algorithms and the hardness constructions.

\begin{lemma}[Maximality is decided by the smallest excluded item]
\label{lem:maxchar}
Let $P\subsetneq I$ be a packing for $c$.  Then $P$ is maximal if and only
if $c-w(P)<\min_{j\in I\setminus P}w_j$; with integral weights,
$c-\min_{j\in I\setminus P}w_j+1\le w(P)\le c$.
\end{lemma}
\begin{proof}
No excluded item fits if and only if the lightest excluded item does not.
\qed
\end{proof}

Sort the items so that $w_1\le\cdots\le w_n$, breaking ties arbitrarily
once and for all, and write $w_{<i}=\sum_{j=1}^{i-1}w_j$ and
$P_{>i}=P\cap\{i+1,\ldots,n\}$.  Fixing item $i$ as the \emph{first}
excluded item in this order forces $1,\ldots,i-1$ into the packing and
leaves a knapsack problem with the two-sided constraint
$\ell_i\le w(P_{>i})\le u_i$, where $\ell_i=c-w_{<i}-w_i+1$,
$u_i=c-w_{<i}$ and $u_i-\ell_i=w_i-1$.  This is
the decomposition behind the dynamic programs of \citet{FLS2015} and
\citet{FLS2017}.

\begin{lemma}[Minimality is decided by the smallest selected item]
\label{lem:minchar}
Let $C\neq\emptyset$ be a cover for $d$.  Then $C$ is minimal if and only if
$w(C)-\min_{j\in C}w_j<d$; with integral weights,
$d\le w(C)\le d+\min_{j\in C}w_j-1$.
\end{lemma}
\begin{proof}
Removing the lightest selected item leaves the largest residual weight.
\qed
\end{proof}

With the same order of nondecreasing weights, fixing the first selected item $i$ of a cover, so that $C\cap\{1,\ldots,i-1\}=\emptyset$, leaves the two-sided constraint $d-w_i\le w(C_{>i})\le d-1$ on $C_{>i}=C\cap\{i+1,\ldots,n\}$,
a window of the same width $w_i-1$.
Figure~\ref{fig:decomposition} displays the two residual windows.

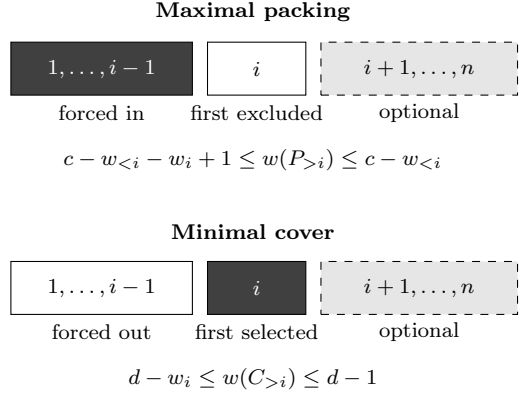
\begin{figure}[htbp]
\centering
\begin{tikzpicture}[
  font=\small,
  selected/.style={draw,fill=black!75,text=white,minimum height=7mm},
  excluded/.style={draw,fill=white,minimum height=7mm},
  optional/.style={draw,dashed,fill=black!10,minimum height=7mm}
]
  % Maximal packing (top panel)
  \node[font=\bfseries] at (3.2,1.75) {Maximal packing};
  \node[selected,minimum width=24mm] at (1.2,1) {$1,\ldots,i-1$};
  \node[excluded,minimum width=13mm] at (3.25,1) {$i$};
  \node[optional,minimum width=26mm] at (5.4,1) {$i+1,\ldots,n$};
  \node[font=\footnotesize] at (1.2,0.42) {forced in};
  \node[font=\footnotesize] at (3.25,0.42) {first excluded};
  \node[font=\footnotesize] at (5.4,0.42) {optional};
  \node at (3.2,-0.2) {$c-w_{<i}-w_i+1\le w(P_{>i})\le c-w_{<i}$};
  % Minimal cover (bottom panel)
  \node[font=\bfseries] at (3.2,-1.15) {Minimal cover};
  \node[excluded,minimum width=24mm] at (1.2,-1.9) {$1,\ldots,i-1$};
  \node[selected,minimum width=13mm] at (3.25,-1.9) {$i$};
  \node[optional,minimum width=26mm] at (5.4,-1.9) {$i+1,\ldots,n$};
  \node[font=\footnotesize] at (1.2,-2.48) {forced out};
  \node[font=\footnotesize] at (3.25,-2.48) {first selected};
  \node[font=\footnotesize] at (5.4,-2.48) {optional};
  \node at (3.2,-3.1) {$d-w_i\le w(C_{>i})\le d-1$};
\end{tikzpicture}
\caption{Two-sided residual constraints after ordering items by nondecreasing
weight and fixing the first excluded or selected item.  Dark boxes are
selected, white solid boxes are excluded, and light dashed boxes are optional.  The
packing panel applies when at least one item is excluded.}
\label{fig:decomposition}
\end{figure}

\begin{corollary}[Exact filling]
\label{cor:exactfill}
\begin{enumerate}[label=(\alph*),leftmargin=2em]
  \item If a maximal packing $P$ for $c$ excludes an item of weight $1$,
        then $w(P)=c$.
  \item If a minimal cover $C$ for $d$ contains an item of weight $1$, then
        $w(C)=d$.
\end{enumerate}
\end{corollary}
\begin{proof}
(a) is Lemma~\ref{lem:maxchar} with $\min_{j\in I\setminus P}w_j=1$.  (b) If $i\in C$
and $w_i=1$, minimality gives $w(C)-1<d$, hence $w(C)\le d$ by integrality.
\qed
\end{proof}

\begin{example}
Let four items have weights $1,2,4,5$, and identify each item by its
weight.  For capacity $c=6$, both $P_1=\{2,4\}$ and $P_2=\{1,4\}$ are
maximal packings: $w(P_1)=6$, while $w(P_2)=5$ and neither excluded item
fits.  Since $w(I)=12$, the paired threshold is $d=w(I)-c=6$.
The complements $C_1=\{1,5\}$ and $C_2=\{2,5\}$ are minimal covers,
of weights $6$ and $7$, respectively.  Thus maximality and minimality
allow slack and excess, while the excluded or selected weight-$1$ item
forces exact filling in $P_1$ and $C_1$.
\end{example}

Corollary~\ref{cor:exactfill} is the mechanism behind every hardness result
of Section~\ref{sec:hard}: a unit-weight item turns the knapsack inequality
into a subset-sum equality.  Note where the unit item sits: outside the
packing, inside the cover.

\subsection{Complementarity}
\label{sec:complement}

The two problems are two views of the same family of sets.  This allows us
to transfer exact algorithms and hardness reductions from one side to the
other, but, as we show next, not approximation guarantees.

\begin{proposition}[Complementarity; \citealp{FLS2017}]
\label{prop:complement}
Let $0<c<w(I)$ and $d=w(I)-c$.  A set $P\subseteq I$ is a maximal
packing for $c$ if and only if $C=I\setminus P$ is a minimal cover for $d$.
Moreover $v(C)=v(I)-v(P)$, hence $\optC=v(I)-\optP$ for the two paired
instances.
\end{proposition}
\begin{proof}
$w(P)\le c$ if and only if $w(C)\ge d$, and for $j\in C$,
$w(P)+w_j>c$ if and only if $w(C\setminus\{j\})<d$.
\qed
\end{proof}

\begin{remark}[Complementation is affine, not multiplicative]
\label{rem:affine}
An $\alpha$-approximation $P$ for \MCMKP gives
$v(I\setminus P)\ge\optC-(\alpha-1)\optP$, an additive guarantee measured
against the wrong optimum.  Example~\ref{ex:affine} shows that no multiplicative guarantee
survives.
Approximation algorithms therefore do not transfer between the two
problems; exact algorithms and reduction gaps do.
\end{remark}
\begin{example}
\label{ex:affine}
Let $m\ge2$, unit costs, $m$ items of weight $1$, $m$ items of weight $m+1$,
one item of weight $m(m+1)$, and $c=m(m+1)+m$.  There are two maximal
packings: the heavy item with the $m$ unit-weight items, of cost $m+1$, and
the $2m$ other items, which weigh exactly $c$, of cost $2m$.  The latter is
a factor-$2$ approximation of $\optP=m+1$, but its complement, the heavy item
alone, is a minimal cover of profit $1$, whereas $\optC=m$.
\end{example}

\subsection{Tools from the literature}

We recall four results that are used later: the approximation algorithms
for the proportional restrictions, which serve as subroutines in
Section~\ref{sec:values}, the exact dynamic program, and
the approximation scheme for the minimum-filling knapsack on which
Section~\ref{sec:fpt} is built.

\begin{theorem}[\citealp{GMP2013}]
\label{thm:gmp-lbp}
The \LBP admits a linear-time $4/3$-approximation and an FPTAS.
\end{theorem}

\begin{theorem}[\citealp{GMP2013}]
\label{thm:maxgreedy}
Max Greedy sorts the items by nonincreasing weight and returns the shortest
prefix that is a cover.  If it returns $t$ items, its weight is at least
$\frac{t}{t+1}\optC$; in particular it is a factor-$2$ approximation for
\GBP, and the factor is tight.
\end{theorem}

\begin{theorem}[\citealp{FLS2017}]
\label{thm:fls}
\MCMKP is weakly $\NP$-hard and is solved in $O(nc)$ time by dynamic
programming.  By complementarity, \MPMKC is weakly $\NP$-hard and is
solved in $O(n(w(I)-d))$ time.  The
decomposition following Lemma~\ref{lem:minchar} also solves \MPMKC directly
in $O(nd)$ time.
\end{theorem}
The weak $\NP$-hardness has been known since \citet[Theorem~7]{GaiZhang2008},
who proved it already for unit costs; the exact dynamic programs came later,
in $O(n^2c)$ time \citep{FLS2015} and then in $O(nc)$ time \citep{FLS2017}.

\begin{theorem}[\citealp{Xu2013}]
\label{thm:xu}
Consider the problem of maximising $v(P)$ subject to
$\beta c\le w(P)\le c$, for a given $\beta\in[0,1)$.  There is an
algorithm that either reports infeasibility or returns
a feasible set of profit at least $(1-\varepsilon)$ times the optimum, in
time $O\bigl(n^2\log n/(1-\beta)+n^2/[\varepsilon(1-\beta)]\bigr)$.
\end{theorem}
For $\beta=1$, deciding whether a feasible set exists is \SubsetSum,
hence $\NP$-complete; this is why Theorem~\ref{thm:xu} requires $\beta<1$.

% ===========================================================================
\section{Inapproximability results}
\label{sec:hard}
% ===========================================================================

We begin with the negative side of the separation: despite
pseudo-polynomial exact algorithms, approximation can be arbitrarily poor.
All results of this section rest on exact filling
(Corollary~\ref{cor:exactfill}) and on reductions from \SubsetSum: given
positive integers $a_1,\ldots,a_{\tilde n}$ and a target $b$, decide whether
some subset sums to exactly $b$.  We assume $a_i\le b$ for $i=1,\ldots,\tilde n$, and
$b<\sum_{i=1}^{\tilde n}a_i$.

\subsection{Gap reductions and complementation}

Every hardness proof below follows the same pattern: a reduction from
\SubsetSum that separates YES and NO instances by a gap in the optimum.  We
state the pattern once, and then explain how a gap proved for \MCMKP yields
one for \MPMKC.

\begin{lemma}[Gap lemma]
\label{lem:gap}
Let $\alpha\ge1$, and suppose a polynomial-time mapping transforms instances of
\SubsetSum to instances of an optimisation problem, together with two
numbers $g_1,g_2$ computable in polynomial time such that $g_2>\alpha\, g_1$.
If the problem is a minimisation and YES instances have $\OPT\le g_1$ while
NO instances have $\OPT\ge g_2$, or if it is a maximisation and YES
instances have $\OPT\ge g_2$ while NO instances have $\OPT\le g_1$, then
the problem admits no polynomial-time $\alpha$-approximation, unless $\PP=\NP$.
\end{lemma}
\begin{proof}
In the minimisation case an $\alpha$-approximation returns at most
$\alpha g_1<g_2$ on YES instances and at least $g_2$ on NO instances; in the
maximisation case it returns at least $g_2/\alpha>g_1$ on YES instances and
at most $g_1$ on NO instances.  Comparing the output with $g_2$, or with
$g_1$, decides \SubsetSum.
\qed
\end{proof}
The factor $\alpha$ may depend on the constructed instance, in
particular on its size.

\paragraph{Complementing a gap.}
Pair the constructed instance with its complement as in
Proposition~\ref{prop:complement}, so that $\optC=v(I)-\optP$.  A gap
$\optP\le g_1$ on YES instances versus $\optP\ge g_2$ on NO instances
becomes the gap $\optC\ge v(I)-g_1$ versus $\optC\le v(I)-g_2$ for the
maximisation problem.  Every hardness result below is proved on the
packing side and carried to the cover side in this way.

\subsection{\vseven{Inapproximability with positive and nonnegative values}}

\vseven{The first result rules out approximation within $2^{p(n)}$ for
every polynomial $p$ even when all values are positive.  Allowing zero
values strengthens this to the absence of any multiplicative guarantee.}
The construction adds a single
item of weight $1$ to a \SubsetSum instance; by exact filling, this item
can be left out only if the capacity is filled exactly.

\begin{theorem}
\label{thm:noconst}
Unless $\PP=\NP$:
\begin{enumerate}[label=(\alph*),leftmargin=2em]
  \item \vseven{neither \MCMKP nor \MPMKC admits a polynomial-time
        $\alpha$-approximation, for any $\alpha\ge1$;}
  \item if all values are required to be positive, neither \MCMKP nor
        \MPMKC admits a polynomial-time $2^{p(n)}$-approximation, for any
        polynomial $p$.
\end{enumerate}
\end{theorem}
\begin{proof}
(a) Given $(a_1,\ldots,a_{\tilde n};b)$, let $c=b$, create an item of
weight $a_i$ and cost $0$ for each $i\in\{1,\ldots,\tilde n\}$, and a
special item $s$ with $w_s=1$ and $v_s=1$.  The instance is nontrivial since
$\sum_{i=1}^{\tilde n}a_i>b$.  If some subset sums to $b$, the
corresponding items fill the knapsack exactly and form a maximal packing of
cost $0$.  A maximal packing that excludes $s$ has weight exactly $b$ by
Corollary~\ref{cor:exactfill}(a) and consists of original items, so it
solves the \SubsetSum instance.  Hence $\optP=0$ on YES instances and
$\optP=1$ on NO instances.  Since $1>\alpha\cdot0$ for every $\alpha$,
Lemma~\ref{lem:gap} applies with $g_1=0$ and $g_2=1$.  On the cover side,
complementation gives $v(I)=1$, so $\optC=1$ on YES instances and $\optC=0$
on NO instances, and Lemma~\ref{lem:gap} applies with the same $g_1$ and
$g_2$.

(b) We may assume that $p$ has only nonnegative integer coefficients, since
increasing $p$ only strengthens the statement.  Take the same items, now
with cost $1$ for each original item and $v_s=(\tilde n+1)\,2^{p(n)}$,
where $n=\tilde n+1$.  The argument of (a) gives $\optP\le\tilde n$ on YES
instances and $\optP\ge v_s>2^{p(n)}\,\tilde n$ on NO instances.  On the
cover side $v(I)=\tilde n+v_s$, so $\optC\ge v_s$ on YES instances and
$\optC\le\tilde n$ on NO instances.  In both cases Lemma~\ref{lem:gap}
applies with $g_1=\tilde n$, $g_2=v_s$ and $\alpha=2^{p(n)}$.  The number
$v_s$ has $O(p(n)+\log\tilde n)$ bits, so the constructed instance has
polynomial size and is built in polynomial time.  As usual, running times
are measured in the binary encoding length, in which arithmetic on numbers
of polynomially many bits takes polynomial time.
\qed
\end{proof}

\subsection{\texorpdfstring{An $n^{1-\varepsilon}$ lower bound for unit costs and profits}{A lower bound for unit costs and profits}}
\label{sec:unitgap}

\vseven{The previous constructions rely on values that differ widely:
an exponentially large cost in part~(b), and zero against positive costs
in part~(a).}
This is not needed: with all costs equal to one the same
mechanism, applied with many unit-weight items instead of one, still
gives a lower bound polynomial in $n$.

\begin{theorem}
\label{thm:unitgap}
For every fixed $0<\varepsilon<1$, unless $\PP=\NP$:
\begin{enumerate}[label=(\alph*),leftmargin=2em]
  \item unit-cost \MCMKP admits no polynomial-time
        $n^{1-\varepsilon}$-approximation;
  \item unit-profit \MPMKC admits no polynomial-time
        $n^{1-\varepsilon}$-approximation.
\end{enumerate}
\end{theorem}
\begin{proof}
(a) Fix an integer $k>1/\varepsilon$, and let $m=\tilde n^k$ and $q=m+1$.  The
instance has capacity $c=qb$, an original item of weight $qa_i$ for each $i\in\{1,\ldots,\tilde n\}$,
and $m$ artificial items of weight $1$; all costs are $1$ and
$n=\tilde n+m$.  On a YES instance the original items of a solution fill the
knapsack exactly, so $\optP\le\tilde n$.  On a NO instance, let $P$ be a
maximal packing that misses an artificial item, let $A$ index its original
items and $t\le m-1$ be the number of its artificial items.  By
Corollary~\ref{cor:exactfill}(a), $q\sum_{i\in A}a_i+t=qb$; reducing modulo
$q>t$ gives $t=0$ and $\sum_{i\in A}a_i=b$, a contradiction.  Hence
$\optP\ge m$ on NO instances.  Since $n\le2\tilde n^k$,
\[
 n^{1-\varepsilon}\,\tilde n\ \le\ 2^{1-\varepsilon}\,\tilde n^{\,k(1-\varepsilon)+1},
\]
which is smaller than $\tilde n^k=m$ for $\tilde n$ larger than a constant
$\tilde N_\varepsilon$, because $k\varepsilon>1$; smaller instances are solved by
enumeration.  For fixed $\varepsilon$ the construction is polynomial, and
Lemma~\ref{lem:gap} with $\alpha=n^{1-\varepsilon}$ applies.

(b) Complementation keeps unit profits and gives $v(I)=n$, so $\optC\ge m$ on
YES instances and $\optC\le\tilde n$ on NO instances; the required
inequality $m>n^{1-\varepsilon}\tilde n$ is the one just proved.
\qed
\end{proof}
As usual for bounds of this form, the degree $k$ depends on $\varepsilon$
only, and each $\varepsilon$ has its own reduction.
The gap mechanism is depicted in Figure~\ref{fig:unitgap-mechanism} and
instantiated numerically in Example~\ref{ex:unitgap}.

\begin{figure}[htbp]
\centering
\begin{tikzpicture}[
  font=\small,
  result/.style={draw,fill=black!10,align=center,minimum width=62mm,
    minimum height=9mm}
]
  \node[font=\bfseries] at (3.3,2.25) {YES: a subset sums to $b$};
  \node at (3.3,1.70) {original items can fill $c=qb$ exactly};
  \draw[thick,->] (0.5,0.75)--(6.5,0.75);
  \draw[thick] (5.7,0.55)--(5.7,0.95);
  \fill (5.7,0.75) circle (2pt);
  \node[font=\footnotesize] at (0.8,0.28) {$0$};
  \node[font=\footnotesize] at (5.7,0.28) {$c$};
  \node[result] at (3.3,-0.55)
    {no unit items needed $\Rightarrow$ cost at most $\tilde n$};

  \begin{scope}[shift={(-7.8,-4.0)}]
  \node[font=\bfseries] at (11.1,2.25) {NO: no subset sums to $b$};
  \node at (11.1,1.70) {feasible original sum is at most $c-q$};
  \draw[thick,->] (8.0,0.75)--(14.4,0.75);
  \draw[thick] (8.9,0.55)--(8.9,0.95);
  \draw[thick] (12.8,0.55)--(12.8,0.95);
  \draw[thick] (13.8,0.55)--(13.8,0.95);
  \draw[->,thick] (8.9,1.18)--(12.8,1.18);
  \node[font=\footnotesize,fill=white,inner sep=1pt] at (10.85,1.18)
    {$+\,m=q-1$ unit items};
  \node[font=\footnotesize] at (8.9,0.28) {$c-q$};
  \node[font=\footnotesize] at (12.8,0.28) {$c-1$};
  \node[font=\footnotesize] at (13.8,0.28) {$c$};
  \node[result] at (11.1,-0.55)
    {every maximal packing includes all $m$ unit items\\
     $\Rightarrow$ cost at least $m$};
  \end{scope}
\end{tikzpicture}
\caption{The unit-cost gap construction with $q=m+1$ and capacity $c=qb$.
On a YES instance original items can fill the capacity.  On a NO instance
their total in any packing is a multiple of $q$ below $c$, and even all
$m=q-1$ unit items leave the load below $c$.  An omitted unit item would
therefore fit, so every maximal packing must include all $m$ of them.}
\label{fig:unitgap-mechanism}
\end{figure}
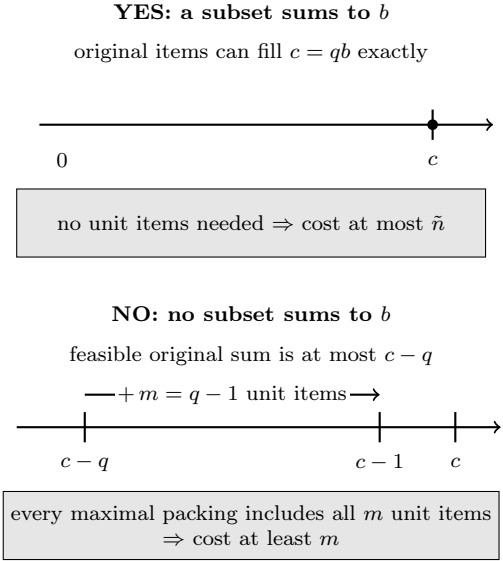

\begin{example}
\label{ex:unitgap}
Let $a=(3,5,7)$ and $k=1$, which is smaller than the proof requires
but keeps the numbers small, so $m=3$ and $q=4$: the weights are
$12,20,28,1,1,1$.  For $b=8$ (YES), $c=32$ and $\{12,20\}$ is a maximal
packing of cost $2$.  For $b=9$ (NO), $c=36$, and a maximal packing avoiding an artificial item
would need $36=4\sigma+t$ with $t\le2$ and $\sigma$ a subset sum of
$\{3,5,7\}$; this forces $t=0$ and $\sigma=9$, which is not a subset sum.
So every maximal packing takes the three artificial items, for instance
$\{28,1,1,1\}$ of cost $4$.
\end{example}

\paragraph{Relation with \citet{GaiZhang2008}.}
The construction of Theorem~\ref{thm:unitgap} is the one used by
\citet[Theorem~7]{GaiZhang2008} for the number-of-jobs objective, where
$\tilde n+1$ tiny jobs of length $1/(\tilde n+1)$ are added; scaling by
$\tilde n+1$ makes them unit-weight items and the original weights multiples
of $\tilde n+1$.  With $\tilde n+1$ tiny items the gap is $\tilde n$ against
$\tilde n+1$, which proves $\NP$-hardness only.  Theorem~\ref{thm:unitgap}
observes that raising the number of tiny items to $\tilde n^k$ costs nothing
and turns the same mechanism into an $n^{1-\varepsilon}$ lower bound, and
that the bound mirrors to the cover side.

\subsection{\texorpdfstring{A matching $n$-approximation}{A matching n-approximation}}

The lower bound of Theorem~\ref{thm:unitgap} is essentially tight, because
in the unit-cost and unit-profit cases every feasible solution is already a good
approximation.

\begin{proposition}
\label{prop:trivial}
Unit-cost \MCMKP and unit-profit \MPMKC admit polynomial-time
$n$-approximations: any maximal packing, found by one greedy insertion
scan, and any minimal cover, found by one greedy deletion scan, are within
factor $n$ of the optimum.
\end{proposition}
\begin{proof}
Inserting the items in any order whenever they fit gives a maximal packing,
since the residual capacity never increases; deleting items from $I$ whenever
the remainder is still a cover gives a minimal cover.  In a nontrivial
instance every maximal packing and every minimal cover is nonempty, so its
value lies between $1$ and $n$.
\qed
\end{proof}

Together, Theorem~\ref{thm:unitgap} and Proposition~\ref{prop:trivial}
essentially determine the polynomial scale of approximability in the unit-cost and unit-profit cases: factor $n$ is trivial, whereas any polynomial improvement in the exponent is impossible.
The bounds do not exclude factors such as $n/\log n$.
With arbitrary values, Proposition~\ref{prop:trivial} does not apply,
since a feasible solution can be arbitrarily far from the optimum, and
\vseven{Theorem~\ref{thm:noconst} excludes every factor $2^{p(n)}$ even
when all values are positive; allowing zero values excludes every
multiplicative factor.}

\begin{table}[ht]
\caption{The mirrored approximation thresholds.}
\label{tab:mirror}
\centering\small\setlength{\tabcolsep}{4pt}
\begin{tabular}{@{}lll@{}}
\toprule
 & Unit-cost \MCMKP & Unit-profit \MPMKC \\
\midrule
Objective & minimise over maximal packings & maximise over minimal covers \\
Gap of the construction & YES: $\le\tilde n$; NO: $\ge m$ & YES: $\ge m$; NO: $\le\tilde n$ \\
Upper bound & $n$ (Prop.~\ref{prop:trivial}) & $n$ (Prop.~\ref{prop:trivial}) \\
Lower bound & no $n^{1-\varepsilon}$ (Thm.~\ref{thm:unitgap}) & no $n^{1-\varepsilon}$ (Thm.~\ref{thm:unitgap}) \\
\bottomrule
\end{tabular}
\end{table}

\subsection{There is no value-indexed dynamic program and no FPTAS}

The hardness results also explain precisely where the analogy with
classical knapsack breaks.  The weight-indexed pseudo-polynomial dynamic
program survives, but a value-indexed analogue would already imply a
polynomial-time algorithm on instances with unit costs or profits.

\begin{corollary}
\label{cor:nofptas}
Unless $\PP=\NP$: (a) even on instances with integer values, neither
\MCMKP nor \MPMKC admits an algorithm whose running time is polynomial in
the binary input size and $v(I)$; (b) neither unit-cost \MCMKP
nor unit-profit \MPMKC admits an FPTAS.
\end{corollary}
\begin{proof}
(a) When all item values are one, $v(I)=n$, and unit-cost \MCMKP is $\NP$-hard
\citep[Theorem~7]{GaiZhang2008}; complementation keeps unit values and gives
$\optC=n-\optP$.  (b) The optimum is an integer in $\{1,\ldots,n\}$.  With
$\varepsilon=1/(2n)$, a scheme returns at most $(1+\varepsilon)\optP<\optP+1$
on the packing side and loses at most
$\optC\,\varepsilon/(1+\varepsilon)<1/2$ on the cover side, so it would be
exact.
\qed
\end{proof}

Part (a) explains why the classical route to an FPTAS for \KP, from the
weight-indexed to a profit-indexed dynamic program and then scaling, is
closed here: when all item values are one, a value-indexed algorithm would be
polynomial.  Part (b) is much weaker than Theorem~\ref{thm:unitgap} and is
recorded only for this reason.

\begin{remark}[Equal weights versus equal values]
\label{rem:unitweights}
With equal weights $w_i=w$ both problems are solved by sorting: a set is a maximal
packing if and only if it has $\lfloor c/w\rfloor$ items, and a minimal
cover if and only if it has $\lceil d/w\rceil$ items.  With weights bounded by a
polynomial in $n$ the dynamic programs are polynomial, since
$c,d\le w(I)\le n\max_{i\in I}w_i$.
However, with equal, or bounded, values
the problems remain $n^{1-\varepsilon}$-inapproximable.  Weights determine
feasibility, values only the objective; what helps on the value side is
correlation with the weights, not magnitude.
\end{remark}

The negative results therefore isolate two sources of difficulty.  Tiny
items create the exact-filling mechanism behind the gap reductions, while
arbitrary values disconnect the objective from the structure of the
feasible sets.  The next two sections examine what happens when these two
features are restricted.

% ===========================================================================
\section{Recovering approximability: restricting the weights}
\label{sec:pos}
% ===========================================================================

The hardness constructions exploit items of vanishing relative weight and
values decoupled from the weights.  This section and the next show that
excluding either feature leads to substantially better algorithmic
behaviour.  We start with the weights: excluding small items gives exact
algorithms and, on the cover side, an approximation scheme polynomial in
the relative size of the smallest item.
Figure~\ref{fig:recovery} shows the two independent restrictions developed
in this section and the next.

\begin{figure}[htbp]
\centering
\begin{tikzpicture}[
  font=\small,
  result/.style={draw,fill=black!10,align=center,minimum width=62mm,
    minimum height=10mm,inner xsep=8pt},
  arrow/.style={->,thick}
]
  \node[font=\bfseries] at (3.25,2.1) {Restrict the weights};

  \draw[arrow] (0.35,1.25)--(6.25,1.25);
  \fill[black!12] (0.35,1.12) rectangle (2.15,1.38);
  \draw[thick] (2.15,1.05)--(2.15,1.45);
  \node[font=\footnotesize] at (1.25,1.7) {tiny items excluded};
  \node[font=\footnotesize] at (2.15,0.85) {$\delta B$};
  \node at (4.4,1.7) {$w_i\ge\delta B$};
  \node[font=\footnotesize] at (3.3,0.45)
    {$B=c$ for packing; $B=d$ for cover};
  \node[result] (weight-result) at (3.25,-0.75)
    {bounded solution size $\Rightarrow$ exact polynomial time\\for fixed $\delta$};
  \draw[arrow] (3.25,0.15)--([yshift=1.5mm]weight-result.north);
  \node[font=\footnotesize,align=center] at (3.25,-1.75)
    {Cover: $w_i-1\ge\delta(d-1)$ gives a scheme\\polynomial in $1/\delta$};

  \begin{scope}[shift={(-7.8,-4.7)}]
  \node[font=\bfseries] at (11.05,2.1) {Restrict the values};
  \draw[arrow] (8.15,1.25)--(14.05,1.25);
  \draw[line width=2pt] (9.35,1.25)--(12.85,1.25);
  \draw[thick] (9.35,1.05)--(9.35,1.45);
  \draw[thick] (12.85,1.05)--(12.85,1.45);
  \node[font=\footnotesize] at (9.35,0.85) {$\lambda_-$};
  \node[font=\footnotesize] at (12.85,0.85) {$\lambda_+$};
  \node at (11.05,1.7) {$v_i/w_i\in[\lambda_-,\lambda_+]$};
  \node[font=\footnotesize] at (11.05,0.45)
    {$\kappa=\lambda_+/\lambda_-$};
  \node[result] (value-result) at (11.05,-0.75)
    {bounded $\kappa$ $\Rightarrow$ $(1+\varepsilon)\kappa$-approximation\\for both problems};
  \draw[arrow] (11.05,0.15)--([yshift=1.5mm]value-result.north);
  \node[font=\footnotesize] at (11.05,-1.75)
    {$\kappa=1$ $\Rightarrow$ FPTAS for both problems};
  \end{scope}
\end{tikzpicture}
\caption{Two ways to recover approximability.  The weight restriction
allows arbitrary costs or profits; the value restriction allows arbitrarily
small items.  The shifted weight condition gives a scheme on the cover side
with polynomial dependence on $1/\delta$.}
\label{fig:recovery}
\end{figure}
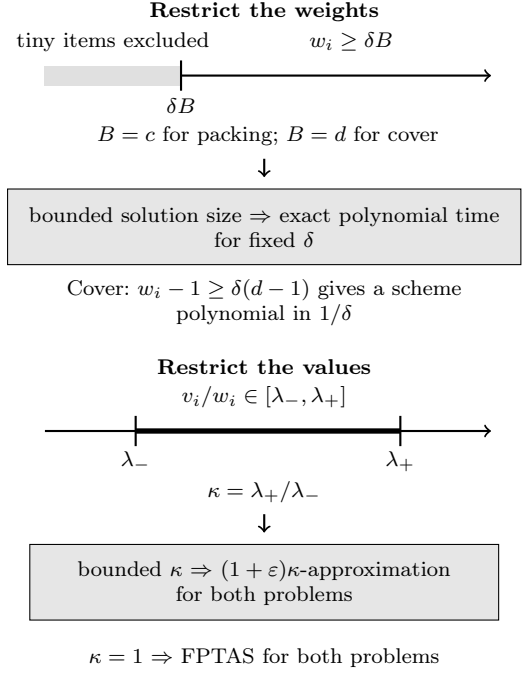

\subsection{Instances without small items}
\label{sec:nonsmall}

The hardness constructions of Section~\ref{sec:hard} use items of
weight~$1$: the special item $s$ of Theorem~\ref{thm:noconst}, where $c=b$,
and the $m$ artificial items of Theorem~\ref{thm:unitgap}, where $c=qb$.
Their weight relative to the capacity is $1/c$, and relative to the paired
threshold it is $1/d$; both vanish as the instances grow.
If every item is instead a fixed fraction of the capacity, or of the
threshold, feasible solutions have bounded cardinality and can be
enumerated.

\begin{theorem}
\label{thm:nonsmall}
Fix $\delta\in(0,1]$.
\begin{enumerate}[label=(\alph*),leftmargin=2em]
  \item If $w_i\ge\delta c$ for every item, every packing has at most
        $\lfloor1/\delta\rfloor$ items and \MCMKP is solved exactly in
        $O(n^{\lfloor1/\delta\rfloor+1})$ time.
  \item If $w_i\ge\delta d$ for every item, every minimal cover has at most
        $\lceil1/\delta\rceil$ items and \MPMKC is solved exactly in
        $O(n^{\lceil1/\delta\rceil+1})$ time.
\end{enumerate}
\end{theorem}
\begin{proof}
(a) A packing $P$ satisfies $|P|\,\delta c\le w(P)\le c$.  (b) A minimal
cover $C$ with lightest weight $\mu\ge\delta d$ satisfies
$(|C|-1)\mu\le w(C)-\mu<d$ by Lemma~\ref{lem:minchar}, so
$|C|-1<1/\delta$.  In
both cases enumerate the subsets of bounded size and test each in $O(n)$
time.
\qed
\end{proof}

For fixed $\delta$ the problems are thus in $\PP$; parameterised by
$1/\delta$ the algorithm is in $\XP$, its exponent depending on the
parameter.

\subsection{An approximation scheme for \MPMKC via the minimum-filling knapsack}
\label{sec:fpt}

The exact algorithm of Theorem~\ref{thm:nonsmall} is polynomial for fixed
$\delta$, but it is not an $\FPT$ algorithm for the parameter $1/\delta$,
which appears in the exponent of $n$.  On the cover side, giving up
exactness yields a running time polynomial in $1/\delta$: once the first
selected item is fixed, the residual problem is exactly the minimum-filling
knapsack of Theorem~\ref{thm:xu}, whose approximation scheme is polynomial
in $1/\delta$.

\begin{proposition}
\label{prop:fptas-delta}
Fix $\delta\in(0,1]$ and suppose $w_i-1\ge\delta(d-1)$ for every item.
Then for every $\varepsilon\in(0,1]$ a minimal cover of profit at least
$\optC/(1+\varepsilon)$ can be computed in time
$O\bigl(\tfrac{n^3}{\delta}(\log n+\tfrac1\varepsilon)\bigr)$, polynomial in
$n$, $1/\varepsilon$ and $1/\delta$.  The profits are arbitrary.
\end{proposition}
\begin{proof}
If $d=1$, every minimal cover is a singleton, so return a most profitable
item.  Hence assume $d\ge2$.
Sort the items in nondecreasing order as before
and try every item $i$ as the
first selected item of an optimal cover.
If $w_i\ge d$, the only minimal
cover with first selected item $i$ is $\{i\}$: every other selected item
has weight at least $w_i$, so any larger set $C$ would have
$w(C)-w_i\ge w_i\ge d$.  Otherwise the minimal covers with first selected
item $i$ are the sets
$\{i\}\cup C_{>i}$ with $C_{>i}\subseteq\{i+1,\ldots,n\}$ and
$d-w_i\le w(C_{>i})\le d-1$.  After discarding the items heavier than $d-1$
this is an instance of Theorem~\ref{thm:xu} with capacity $d-1$ and
$\beta_i=(d-w_i)/(d-1)$, where $1/(1-\beta_i)=(d-1)/(w_i-1)\le1/\delta$.
Run the algorithm of Theorem~\ref{thm:xu} with accuracy
$\eta=\varepsilon/(1+\varepsilon)$.  If it reports infeasibility, discard
$i$; this does not happen for the first selected item of an optimal cover.
Otherwise it returns $C_{>i}$ with $d-w_i\le w(C_{>i})\le d-1$, so
$\{i\}\cup C_{>i}$ is a minimal cover by the decomposition following
Lemma~\ref{lem:minchar}, and
$v_i+v(C_{>i})\ge v_i+(1-\eta)\OPT_i
\ge(1-\eta)(v_i+\OPT_i)=(v_i+\OPT_i)/(1+\varepsilon)$, where
$\OPT_i$ is the optimum of the residual problem.  Return the best set over
all $i$.  The second inequality uses $v_i\ge0$.  Since
$1/\eta\le2/\varepsilon$, the running time is $n$ times the
bound of Theorem~\ref{thm:xu}.
\qed
\end{proof}

\begin{corollary}
\label{cor:fptas-fixed-delta}
For every fixed $\delta\in(0,1]$, \MPMKC restricted to instances with
$w_i-1\ge\delta(d-1)$ for all $i\in I$ admits an FPTAS.  More generally, if
$w_{\min}=\min_{i\in I}w_i\ge2$, the scheme is an FPTAS on every family of
instances on which $(d-1)/(w_{\min}-1)$ is bounded by a polynomial in the
input size.
\end{corollary}
\begin{proof}
The first statement is Proposition~\ref{prop:fptas-delta} with $1/\delta$
constant.  For the second, the proof of Proposition~\ref{prop:fptas-delta}
spends, on a candidate $i$ with $w_i<d$, time polynomial in $n$,
$1/\varepsilon$ and $(d-1)/(w_i-1)\le(d-1)/(w_{\min}-1)$.
\qed
\end{proof}
The bound involves the lightest item because every item is a
candidate first selected item: an optimal cover whose first selected item
has weight $1$ leads to an exact subset-sum problem by
Corollary~\ref{cor:exactfill}(b).  As $d$ is encoded in binary, a polynomial
bound on $(d-1)/(w_{\min}-1)$ is a genuine restriction; it is also the form
in which \citet{Xu2013} states his own result.

We add three comments.  First, for fixed $\delta$
Theorem~\ref{thm:nonsmall}(b) already solves the problem exactly, but with
$1/\delta$ in the exponent of $n$; here it is a polynomial factor, so the
scheme survives a $\delta$ that shrinks with the instance.  Second, the
excluded case $w_i=1$ is $\beta_i=1$, the boundary at which feasibility
becomes \SubsetSum and Theorem~\ref{thm:xu} no longer applies, and it is
also the exact-filling boundary
of Corollary~\ref{cor:exactfill}(b); the two hard boundaries coincide.
Third, there is no packing-side analogue by this route: the residual
packing problem minimises, whereas Theorem~\ref{thm:xu} maximises, and the
minimisation knapsack FPTAS
handles a one-sided load constraint only.

% ===========================================================================
\section{Recovering approximability: restricting the values}
\label{sec:values}
% ===========================================================================

We now turn to the values.  We first treat the proportional endpoint
$v_i=w_i$, where both problems admit approximation schemes, and then show
how much of this survives when the values deviate from the weights by a
bounded factor.

\subsection{The proportional case: approximation schemes for the Greedy Boss Problem}
\label{sec:gbp}

In the proportional case the packing problem already admits an FPTAS
(Theorem~\ref{thm:gmp-lbp}), whereas for the cover problem only the tight
factor $2$ of Max Greedy was known (Theorem~\ref{thm:maxgreedy}).
We now show that the cover side also
admits approximation schemes: a simple combinatorial PTAS, and then an
FPTAS.  Throughout this subsection $v_i=w_i$, so $\optC$ is the
maximum weight of a minimal cover.
For the PTAS developed first, the items are now sorted by nonincreasing weight $w_1\ge\cdots\ge w_n$.

Max Greedy may err at its first choice: for an integer $N\ge3$, on $d=N$ and
$w=(N,N-1,N-1)$ it takes the heaviest item and stops, while $\{2,3\}$ is a
minimal cover of weight $2N-2$.
Algorithm $\mathcal{A}_k$ guesses the $k$ heaviest items of an optimal cover and lets Max Greedy complete the guess with items no heavier than the lightest item among the guessed ones.

\paragraph{Algorithm $\mathcal{A}_k$.}
Let $k\ge0$ be an integer.  For every $T\subseteq I$ with $|T|\le k$:
\begin{enumerate}[label=(\roman*),leftmargin=2.5em]
  \item if $w(T)\ge d$: record $T$ if it is a minimal cover, otherwise
        discard the guess;
  \item if $w(T)<d$: let $\mu(T)=\min_{j\in T}w_j$ (with
        $\mu(\emptyset)=+\infty$) and
        $R(T)=\{i\in I\setminus T: w_i\le\mu(T)\}$.  If $w(T)+w(R(T))<d$,
        discard the guess; otherwise run Max Greedy on $R(T)$ with
        threshold $d-w(T)$, call its output $G(T)$, and record
        $T\cup G(T)$.
\end{enumerate}
Return a recorded set of maximum weight.  For $k=0$ the algorithm is Max
Greedy.  After sorting, it runs in $O(n^{k+1})$ time for $k\ge1$.

\begin{example}
On $d=N$, $w=(N,N-1,N-1)$, the guess $T=\{2\}$ has $R(T)=\{3\}$, and
Max Greedy with threshold $1$ returns $\{3\}$: $\mathcal{A}_1$ records the
optimum $\{2,3\}$.
\end{example}

We first show that the algorithm only records feasible solutions, and then
bound the weight of the best one.

\begin{lemma}
\label{lem:akfeasible}
Every set recorded by $\mathcal{A}_k$ is a minimal cover for $d$.
\end{lemma}
\begin{proof}
Sets recorded in (i) are minimal covers by the test.  Let $C=T\cup G(T)$ be
recorded in (ii).  Since the guess is not discarded,
$w(G(T))\ge d-w(T)$ and $C$ is a cover.  Let $g$ be the last item added by
Max Greedy; its stopping rule gives $w(C)-w_g<d$.  Every $j\in C$ has
$w_j\ge w_g$: for $j\in T$ because $g\in R(T)$, and for $j\in G(T)$ because
Max Greedy scans by nonincreasing weight.  Hence $w(C\setminus\{j\})<d$ for
every $j\in C$.
\qed
\end{proof}

\begin{theorem}
\label{thm:gbpk}
For every integer $k\ge0$, $\mathcal{A}_k$ returns a minimal cover $C$ with
$w(C)\ge\frac{k+1}{k+2}\optC$; that is, $\mathcal{A}_k$ is a polynomial-time
$\frac{k+2}{k+1}$-approximation for \GBP.
\end{theorem}
\begin{proof}
Let $C^*$ be an optimal minimal cover and $z^*=w(C^*)$.  If $|C^*|\le k$,
the guess $T=C^*$ is recorded in (i).  Otherwise let $T$ consist of the $k$
heaviest items of $C^*$ and let $i^*\in C^*\setminus T$ be a lightest item
of $C^*$.  As a proper subset of a minimal cover, $T$ has $w(T)<d$; and
$C^*\setminus T\subseteq R(T)$, so $w(T)+w(R(T))\ge z^*\ge d$.  Hence
$C=T\cup G(T)$ is recorded.  The first item chosen by Max Greedy is a
heaviest item of $R(T)\ni i^*$, and the items of $T$ weigh at least
$w_{i^*}$, so $w(C)\ge(k+1)w_{i^*}$.  If $w_{i^*}\ge z^*/(k+2)$ we are done.
Otherwise Lemma~\ref{lem:minchar} applied to $C^*$ gives
$z^*<d+w_{i^*}<d+z^*/(k+2)$, that is, $\frac{k+1}{k+2}z^*<d\le w(C)$.
\qed
\end{proof}

\begin{corollary}
\label{cor:ptas}
\GBP admits a PTAS: for every fixed $\varepsilon>0$, $\mathcal{A}_k$ with
$k=\lceil1/\varepsilon\rceil-1$ is a $(1+\varepsilon)$-approximation running
in $n^{O(1/\varepsilon)}$ time.
\end{corollary}
\begin{proof}
$k+1\ge1/\varepsilon$ gives $\frac{k+2}{k+1}\le1+\varepsilon$, and the
running time is $O(n^{k+1})$.
\qed
\end{proof}

\begin{remark}
\label{rem:gaizhang-ak}
The guessing technique is the one of \citet[Theorem~3]{GaiZhang2008} for the
makespan objective of the Lazy Bureaucrat Problem with a common deadline,
and more closely of \citet[Algorithm~3 and Theorem~2]{GMP2014}, who guess
the $k$ heaviest elements of an optimal lazy matroid solution and restrict
the completion to lighter elements.  $\mathcal{A}_k$ is its cover-side
counterpart on the free matroid; the restriction to items no heavier than
the lightest guess is what keeps the output minimal
(Lemma~\ref{lem:akfeasible}).
\end{remark}

The scheme $\mathcal{A}_k$ is combinatorial and strongly polynomial for every fixed accuracy, but its running time is $n^{O(1/\varepsilon)}$.
The pseudo-polynomial property of \GBP can be exploited further:
combining the decomposition by the first selected item with an FPTAS for the optimisation version of \SubsetSum (maximise $w(S)$ subject to $w(S)\le b$) yields a fully polynomial-time approximation scheme.
Such FPTASs are classical \citep{IbarraKim1975,KMPS2003,KPP2004}:
for every $\eta\in(0,1)$ they return a set $S$ feasible for \SubsetSum with $w(S)\ge(1-\eta)\max\{w(S'):w(S')\le b\}$, in time polynomial in the number of items and $1/\eta$.

\paragraph{Algorithm $\mathcal{F}_\varepsilon$.}
Let $\varepsilon\in(0,1]$ and $\eta=\varepsilon/(1+\varepsilon)$.  Compute a
minimal cover $C^0$ by deleting items from $I$ while the remainder is still
a cover, and record it.  Sort the items as in Section~\ref{sec:struct}, by
nondecreasing weight with the fixed tie-breaking order.  For every item $i$:
if $w_i\ge d$, record $\{i\}$; otherwise run a \SubsetSum FPTAS with
accuracy $\eta$ on the items $i+1,\ldots,n$ with bound $d-1$, let $S_i$ be
the returned set, and record $\{i\}\cup S_i$ if $w(S_i)\ge d-w_i$.  Return a
recorded set of maximum weight.

\begin{theorem}
\label{thm:gbpfptas}
For every $\varepsilon\in(0,1]$, $\mathcal{F}_\varepsilon$ returns a minimal
cover $C$ with $w(C)\ge\optC/(1+\varepsilon)$, in time $n$ times the running
time of the \SubsetSum FPTAS with accuracy $\eta$.  Hence \GBP admits an
FPTAS.
\end{theorem}
\begin{proof}
\emph{Feasibility.}  $C^0$ is a minimal cover, and so is $\{i\}$ when
$w_i\ge d$.  A set $\{i\}\cup S_i$ is recorded only if
$d-w_i\le w(S_i)\le d-1$ with $S_i\subseteq\{i+1,\ldots,n\}$, so it is a
minimal cover by the decomposition following Lemma~\ref{lem:minchar}.

\emph{Guarantee.}  Let $C^*$ be an optimal minimal cover and $z^*=w(C^*)$.
If $z^*\le(1+\varepsilon)d$, then $w(C^0)\ge d\ge z^*/(1+\varepsilon)$.
Otherwise let $i$ be the first selected item of $C^*$.  If $w_i\ge d$, then
$C^*=\{i\}$ is recorded.  If $w_i<d$, let $S^*=C^*\setminus\{i\}$ and
$s^*=w(S^*)$, so that $d-w_i\le s^*\le d-1$; in particular $S^*$ is
feasible for the residual \SubsetSum instance.  It is optimal there:
a heavier feasible residual set would, by the decomposition following
Lemma~\ref{lem:minchar}, give a minimal cover heavier than $C^*$.
Hence $(1-\eta)s^*\le w(S_i)\le s^*$.  Since
$s^*-(d-w_i)=z^*-d>\varepsilon d$ and $\eta s^*<\eta d<\varepsilon d$, we get
$w(S_i)\ge s^*-\eta s^*>d-w_i$, so $\{i\}\cup S_i$ is recorded, and
\[
 w_i+w(S_i)\ \ge\ w_i+(1-\eta)s^*\ \ge\ (1-\eta)z^*\ =\ \frac{z^*}{1+\varepsilon}.
\]
The running time is dominated by the $n$ calls to the \SubsetSum FPTAS;
since $1/\eta\le2/\varepsilon$, it is polynomial in the input size and $1/\varepsilon$.
\qed
\end{proof}

Figure~\ref{fig:gbp-margin} illustrates why the approximate residual remains above the lower threshold in the high-optimum case of the proof.
\begin{figure}[htbp]
\centering
\begin{tikzpicture}[font=\small,x=0.78cm]
  \node[font=\bfseries] at (7.25,1.55)
    {high-optimum case: $z^*>(1+\varepsilon)d$};
  \draw[thick,->] (0.65,0.15)--(14.55,0.15);
  \foreach \x in {1.3,5.6,9.4,13.4}
    \draw[thick] (\x,-0.05)--(\x,0.35);
  \node at (1.3,-0.28) {$d-w_i$};
  \node at (5.6,-0.28) {$w(S_i)$};
  \node at (9.4,-0.28) {$s^*$};
  \node at (13.4,-0.28) {$d-1$};

  \draw[<->,thick] (1.3,0.72)--(9.4,0.72);
  \node[fill=white,inner sep=1pt] at (5.35,1.02)
    {$s^*-(d-w_i)=z^*-d>\varepsilon d$};
  \draw[<->,thick] (5.6,-0.80)--(9.4,-0.80);
  \node[fill=white,inner sep=1pt] at (7.5,-1.13)
    {$s^*-w(S_i)\le\eta s^*<\varepsilon d$};
  \node at (7.25,-1.72)
    {$d-w_i<w(S_i)\le d-1$};
\end{tikzpicture}
\caption{Feasibility margin in the FPTAS for \GBP.  For the first selected
item $i$ of an optimal cover, $z^*$ is its total weight and $s^*=z^*-w_i$
is its residual weight.  In the high-optimum case, the residual optimum
lies more than $\varepsilon d$ above the lower threshold, while the FPTAS for \SubsetSum loses less than $\varepsilon d$.  Positions are schematic.}
\label{fig:gbp-margin}
\end{figure}
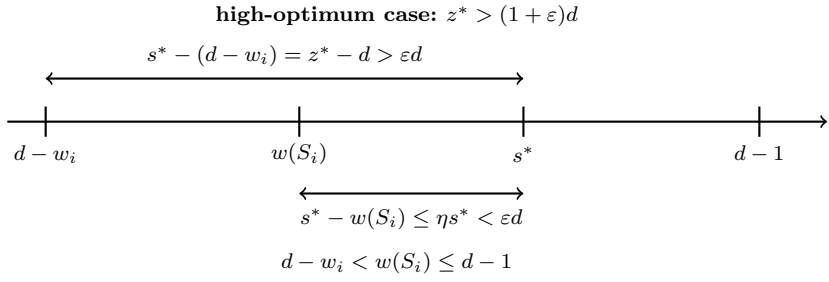

Using the \SubsetSum FPTAS of \citet{KMPS2003}, the overall running time is
polynomial in the input size and $1/\varepsilon$.  Theorem~\ref{thm:gbpfptas} settles the
approximability question raised by \citet{GMP2013}: at the proportional
endpoint both the packing and the cover problem admit an FPTAS.

\subsection{Bounded value-to-weight spread}
\label{sec:spread}

Having approximation schemes at the proportional endpoint $\kappa=1$, we now
ask how much of their guarantee survives when values are allowed to deviate
from the weights by a bounded factor.  Feasibility depends on the weights only, so an
algorithm for the proportional problem can be run on the weights and its
output re-evaluated with the true values.

\begin{theorem}
\label{thm:spread}
Let $\kappa$ be the spread defined in \eqref{eq:kappa}.
For every $\varepsilon>0$:
\begin{enumerate}[label=(\alph*),leftmargin=2em]
  \item \MCMKP admits a $(1+\varepsilon)\kappa$-approximation;
  \item \MPMKC admits a $(1+\varepsilon)\kappa$-approximation;
\end{enumerate}
in both cases with running time polynomial in the input size and $1/\varepsilon$.
\end{theorem}
\begin{proof}
It suffices to consider $0<\varepsilon\le1$: for $\varepsilon>1$, running
the algorithms with accuracy $1$ gives a stronger guarantee.
(a) Let $P_\varepsilon$ be the maximal packing returned by the FPTAS of
Theorem~\ref{thm:gmp-lbp} for the objective $w$, and let $P_w^*$ and $P_v^*$
be maximal packings of minimum weight and of minimum cost.  Then
\[
 v(P_\varepsilon)\le\lambda_+w(P_\varepsilon)
 \le(1+\varepsilon)\lambda_+w(P_w^*)
 \le(1+\varepsilon)\lambda_+w(P_v^*)
 \le(1+\varepsilon)\kappa\,v(P_v^*).
\]
(b) Symmetrically, let $C_\varepsilon$ be the minimal cover returned by
the FPTAS of Theorem~\ref{thm:gbpfptas} for the objective $w$, and let $C_w^*$
and $C_v^*$ be minimal covers of maximum weight and of maximum profit.  Then
\[
 v(C_\varepsilon)\ge\lambda_-w(C_\varepsilon)
 \ge\frac{\lambda_-}{1+\varepsilon}w(C_w^*)
 \ge\frac{\lambda_-}{1+\varepsilon}w(C_v^*)
 \ge\frac{v(C_v^*)}{(1+\varepsilon)\kappa}.
\]
\qed
\end{proof}

\begin{remark}
\label{rem:linearspread}
Any approximation for the proportional problem can be used instead, at the
price of its factor: Max Greedy gives a linear-time $2\kappa$-approximation
on the cover side and the $4/3$-approximation of Theorem~\ref{thm:gmp-lbp}
gives $\frac43\kappa$ on the packing side.
\end{remark}

Both endpoints of the spread parameter are covered: $\kappa$ unbounded is
Theorem~\ref{thm:noconst}(b), whose instances have $\kappa$ exponential
in $n$, and $\kappa=1$, where the packing problem has an
FPTAS and so has the cover problem (Theorem~\ref{thm:gbpfptas}), is the
proportional endpoint.  The algorithms of Theorem~\ref{thm:spread} are not
FPTASs for $\kappa>1$, since their ratio tends to $\kappa$ and not to $1$, as $\varepsilon$ goes to $0$.

% ===========================================================================
\section{Conclusions and open problems}
\label{sec:concl}
% ===========================================================================

Maximal knapsack packings and minimal knapsack covers are weakly
$\NP$-hard and solved by weight-indexed pseudo-polynomial dynamic programs,
exactly as the knapsack problem, but they are much harder to approximate:
\KP admits an FPTAS, whereas \vseven{\MCMKP and \MPMKC admit no
$2^{p(n)}$-approximation with positive values for any polynomial $p$,
and no multiplicative factor if zero values are allowed.}
\vseven{They admit no $n^{1-\varepsilon}$-approximation even with unit}
costs or profits (Table~\ref{tab:comparison}).  Approximability can be
recovered by restricting the input data (Table~\ref{tab:landscape}): on the
weights, by excluding small items; on the values, by bounding the
value-to-weight spread; at the proportional endpoint both problems admit
FPTASs.  Thus inclusion-wise extremality does not by itself
destroy approximability; our lower-bound constructions exploit tiny items
together with objectives that are not proportional to the weights.

\begin{table}[ht]
\caption{Polynomial-time approximability of the two problems, assuming
$\PP\neq\NP$.}
\label{tab:landscape}
\centering\small
\begin{tabular}{@{}L{0.24\linewidth}L{0.32\linewidth}L{0.36\linewidth}@{}}
\toprule
Regime & \MCMKP & \MPMKC \\
\midrule
Proportional ($\kappa=1$) & FPTAS \citep{GMP2013} & FPTAS (Thm.~\ref{thm:gbpfptas}) \\
Bounded spread $\kappa$ & $(1+\varepsilon)\kappa$ (Thm.~\ref{thm:spread}) & $(1+\varepsilon)\kappa$ (Thm.~\ref{thm:spread}) \\
Non-small items ($w_i\ge\delta c$, resp.\ $w_i\ge\delta d$), $\delta$ fixed & exact, $n^{O(1/\delta)}$ (Thm.~\ref{thm:nonsmall}) & exact, $n^{O(1/\delta)}$ (Thm.~\ref{thm:nonsmall}) \\
$w_i-1\ge\delta(d-1)$ & --- & scheme polynomial in $1/\delta$; FPTAS for fixed $\delta$ (Prop.~\ref{prop:fptas-delta}, Cor.~\ref{cor:fptas-fixed-delta}) \\
Unit costs/profits & $n$ yes, $n^{1-\varepsilon}$ no (Thm.~\ref{thm:unitgap}, Prop.~\ref{prop:trivial}) & $n$ yes, $n^{1-\varepsilon}$ no (Thm.~\ref{thm:unitgap}, Prop.~\ref{prop:trivial}) \\
Arbitrary objective & \vseven{no $2^{p(n)}$ with positive costs; no factor if zero is allowed} (Thm.~\ref{thm:noconst}) & \vseven{no $2^{p(n)}$ with positive profits; no factor if zero is allowed} (Thm.~\ref{thm:noconst}) \\
\bottomrule
\end{tabular}
\end{table}

\paragraph{Open problems.}
The tractable cases above are sufficient conditions, not a complete
classification.
The remaining questions concern the boundary between hard and tractable instances.
\begin{enumerate}[leftmargin=2em,itemsep=2pt]
  \item \emph{Tight approximation with unit costs or profits.}  The bounds of
        Section~\ref{sec:hard} determine the polynomial exponent but leave a
        subpolynomial gap.  Does either problem admit an
        $o(n)$-approximation, or can the lower bound be strengthened to
        $n^{1-o(1)}$?
  \item \emph{Fixed spread.}  At $\kappa=1$ both problems admit FPTASs,
        whereas for unbounded $\kappa$ no factor $2^{p(n)}$ is achievable
        (Theorem~\ref{thm:noconst}(b)).
        For a fixed $\kappa>1$, is there a PTAS, or are the problems
        $\APX$-hard?  The neighbouring model $v_i=w_i+\text{const}$ is the
        strongly correlated class of \citet{KPP2004}.
  \item \emph{Two-sided minimisation with relative slack.}  For minimising
        $v(S)$ subject to $\ell\le w(S)\le u$, even deciding feasibility is
        $\NP$-hard in general, since $\ell=u$ encodes \SubsetSum.  Under the
        promise $u-\ell\ge\delta u$, is there an approximation scheme
        polynomial in $1/\delta$, analogous to the scheme of \citet{Xu2013}
        for the maximisation version?  It would give the packing-side
        analogue of Proposition~\ref{prop:fptas-delta}.
  \item \emph{Parameterised complexity of the no-small-items regime.}  The
        exact algorithm of Theorem~\ref{thm:nonsmall} runs in
        $n^{O(1/\delta)}$ time, hence is in $\XP$ for the parameter
        $1/\delta$.  Is there an $\FPT$ algorithm, or are the problems
        $\Wone$-hard in this parameter?
\end{enumerate}

\section*{Use of AI tools}
The authors used a large language model (Claude, Anthropic) as an assistant
in preparing this manuscript: to restructure and edit the text, to rewrite
proofs from the authors' earlier technical notes in a more concise form, and
to assist in checking the literature and the proofs of results used from
other papers.  All results, proofs and references were verified by the
authors, who take full responsibility for the content of the paper.

\bibliographystyle{spbasic}
\bibliography{biblio}

\section*{Statements and Declarations}

\paragraph{Funding.}
The authors did not receive support from any organization for the
submitted work.

\paragraph{Competing interests.}
The authors have no competing interests to declare that are relevant to
the content of this article.

\paragraph{Data availability.}
Data sharing is not applicable to this article, as no datasets were
generated or analysed during the current study.

\end{document}